\documentclass{sigplanconf}

\newtheorem{prop}{Proposition}
\newcommand{\Mscr}{{\cal M}}

\newcommand{\qed}{\nobreak\hfill\nobreak\quad\nobreak\rule{5pt}{7pt}}
\newenvironment{proof}{\noindent {\bf Proof}\quad}{\qed}
\newcount\PLv\newcount\PLw\newcount\PLx\newcount\PLy\newdimen\PLyy\newdimen\PLX
\newbox\PLdot \setbox\PLdot\hbox{\tiny.} 
\def\scl{.06} 

\def\PLot#1{\PLx`#1\advance\PLx-42\PLy\PLx\PLv\PLx\divide\PLy9\PLw\PLy\multiply
\PLw9\advance\PLx-\PLw\advance\PLx-4\PLy-\PLy\advance\PLy4\PLX=\the\PLx pt
\advance\PLyy\the\PLy pt\wd\PLdot=\scl\PLX\raise\scl\PLyy\copy\PLdot}
\def\draw#1{\ifx#1\end\let\next=\relax\else\PLot#1\let\next=\draw\fi\next}

\def\invamp{\hbox{\PLyy=70pt\draw :::;DMV_gqppyyyyyooooxxxnnwvlutkjaWNE=5-./9
9:::CCCC:::99/..--544=EENWWaajjjkktttttttNNNVVVVVVVV\end \hskip4pt}}

\newbox\iabox\setbox\iabox\invamp \def\Invamp{\copy\iabox}

\newcommand{\lpar}{\mathrel{\Invamp}}
\newcommand{\limp}{\mathbin{-\hspace{-0.70mm}\circ}}
\newcommand{\bla}{ \mathrel{\mbox{$\circ\!-$}}}
\newcommand{\lequiv}{\mathbin{\circ\hspace{-1.30mm}-\hspace{-1.30mm}\circ}}
\newcommand{\with}{\mathbin{\&}}
\newcommand{\bang}{\mathop{!}}

\newcommand{\llitem}[1]{\hbox{\em item}~#1}
\newcommand{\mssub}{\sqsubseteq}
\newcommand{\mseq }{\mathrel{\buildrel m\over =}}
\newcommand{\msmodels}{\models_m}
\newcommand{\Rho}{\mathrel{\rho}}
\newcommand{\hatRho}{\mathrel{\hat\rho}}

\def\relbar{\mathrel{\smash-}}
\def\joinrelm{\mathrel{\mkern-3mu}}
\def\tailpiece{\kern 1pt\vrule height 1ex width 0.3ex depth -.3ex}
\def\seqsym{\mathrel{\tailpiece\joinrelm\relbar}}

\newcommand{\tseqsym}{\mathrel{\kern3pt\seqsym\kern-10pt\seqsym}}

\newcommand{\twoseq}[2]{#1 \seqsym #2}
\newcommand{\quest}{\mathop{?}}
\newcommand{\ot}{\otimes}
\newcommand{\iimp}{\Rightarrow}
\newcommand{\zero}{{\bf 0}}
\newcommand{\ra}{\rightarrow}

\newcommand{\Tok}[1]{\hbox{\tt #1}}
\newcommand{\Bottom}{\bot}
\newcommand{\One}{{\bf 1}}
\newcommand{\Allx}[2]{\forall #1.#2}

\newcommand{\With}[2]{#1\\ #2}
\newcommand{\Impliedby}[2]{#2\iimp #1}
\newcommand{\Llpar}[2]{#1\lpar #2}

\newcommand{\Llequiv}[2]{(#1\lequiv #2)}

\newcommand{\Llitem}[1]{\hbox{\em item}~#1}

\begin{document}
\conferenceinfo{PPDP'06} {July 10--12, 2006, Venice, Italy.}
\copyrightyear{2006}
\copyrightdata{1-59593-388-3/06/0007} 

\preprintfooter{Collection Analysis for Horn Clause Programs} 

\title{Collection Analysis for Horn Clause Programs} 
\subtitle{[Extended Abstract]}

\authorinfo{Dale Miller}
           {INRIA \& LIX, \`Ecole Polytechnique, Rue de Saclay\\
            91128 Palaiseau, France}
           {dale.miller [at] inria.fr}

\maketitle
\begin{abstract}
We consider approximating data structures with collections of the
items that they contain.  For examples, lists, binary trees, tuples,
etc, can be approximated by sets or multisets of the items within
them.  Such approximations can be used to provide partial correctness
properties of logic programs.  For example, one might wish to specify
than whenever the atom $sort(t,s)$ is proved then the two lists $t$
and $s$ contain the same multiset of items (that is, $s$ is a
permutation of $t$).  If sorting removes duplicates, then one would
like to infer that the sets of items underlying $t$ and $s$ are the
same.  Such results could be useful to have if they can be determined
statically and automatically.  We present a scheme by which such
collection analysis can be structured and automated.  Central to this
scheme is the use of linear logic as a computational logic underlying
the logic of Horn clauses.
\end{abstract}


            

\section{Introduction}

Static analysis of logic programs can provide useful information for
programmers and compilers.  Typing systems, such as in $\lambda$Prolog
\cite{nadathur88iclp,nadathur92types}, have proved valuable during the
development of code: type errors often represent program errors that
are caught at compile time when they are easier to find and fix than
at runtime when they are much harder to repair.  Static type
information also provides valuable documentation of code since it
provides a concise approximation to what the code does.

In this paper we describe a method by which it is possible to infer
that certain relationships concerning collections underlying
structured data hold.  We shall focus on relations that are also
decidable and can be done during compile time analysis of logic
programs.   We shall use {\em multisets} and {\em sets} to {\em
approximate} more complicated structures as lists and binary trees.
Consider, for example, a list sorting program that maintains
duplicates of elements.  Part of the correctness of a sort program
includes the fact that if the atomic formula $sort(t,s)$ is provable,
then $s$ is a permutation of $t$ that is in-order.  The proof of such
a property is likely to involve inductive arguments requiring the
invention of invariants: in other words, this is not likely to be a
property that can be inferred statically during compile time.  On the
other hand, if the lists $t$ and $s$ are approximated by multisets
(that is, if we forget the order of items in lists), then it might be
possible to establish that if the atomic formula $sort(t,s)$ is provable,
then the multiset associated to $s$ is equal to the multiset associated
to $t$.  If that is so, then it is immediate that the lists $t$ and
$s$ are, in fact, permutations of one another (in other words, no
elements were dropped, duplicated, or created during sorting).  
As we shall see, such properties based on using multisets to
approximate lists  can often be done statically. 

This paper considers exclusively the static analysis of first-order
Horn clauses but it does so by making substitution instances of such
Horn clauses that carry them into linear logic.  Proofs for the
resulting linear logic formulas are then attempted as part of static
analysis.

\section{The undercurrents}\label{undercurrents}

There are various themes that underlie our approach to inferring
properties of Horn clause programs.  We list them explicitly below.
The rest of the paper can be seen as a particular example of how these
themes can be developed.

\subsection{If typing is important, why use only one?}

Types and other static properties of programming languages have proved
important on a number of levels.  Typing can be useful for
programmers: they can offer important invariants and document for
code.  Static analysis can also be used by compilers to uncover useful
structures that allow compilers to make choices that can improve
execution.  While compilers might make use of multiple static analysis
regimes, programmers do not usually have convenient access to multiple
static analyzes for the code that they are composing.  Sometimes, a
programming language provides no static analysis, as is the case with
Lisp and Prolog.  Other programming languages offer exactly one typing
discipline, such as the polymorphic typing disciplines of Standard ML
and $\lambda$Prolog (SML also statically determines if a given
function defined over concrete data structures cover all possible
input values).  It seems clear, however, that such analysis of
code, if it can be done quickly and incrementally, might have
significant benefits for programmers during the process of writing code.
For example, a programmer might find it valuable to know that a
recursive program that she has just written has linear or quadratic
runtime complexity,
or that a relation she just specified actually defines a function.
The Ciao system preprocessor \cite{hermenegildo05scp} provides for
such functionality by allowing a programmer to write various
properties about code that the preprocessor attempts to verify.
Having an open set of properties and analysis tools is an
interesting direction for the design of a programming language.  The
collection analysis we discuss here could be just one such analysis
tool.

\subsection{Logic programs as untyped $\lambda$-expressions}

If we do not commit to just one typing discipline, then it seems
sensible to use a completely untyped setting for encoding programs and
declarations.  Given that untyped $\lambda$-terms provide for
arbitrary applications and arbitrary abstractions, such terms can
provide an appealing setting for the encoding of program expressions,
type expressions, assertions, invariants, etc.  Via the well
developed theory of $\lambda$-conversion, such abstractions can be
instantiated with a variety of other objects.  Abstractions can be
used to encode quantifiers within formulas as well as binding
declarations surrounding entire programs.

In logic programming, proofs can be viewed as computation traces and
such proof objects can also be encoded as untyped $\lambda$-terms.
Instantiations into proofs is also well understood since it is closely
related to the elimination of cut in sequent calculus or to
normalization in natural deduction proofs.  The fact that proofs and
programs can be related simply in a setting where substitution into
both has well understood properties is certainly one of the strengths
of the proof theoretic foundations of logic programming (see, for
example, \cite{miller91apal}).

\subsection{What good are atomic formulas?}

In proof theory, there is interesting problem of duality involving
atomic formulas.   The {\em initial rule} and the {\em cut rule}
given as
$$\frac{}{\twoseq{C}{C}}
  \hbox{\ Initial} \qquad
  \frac{\twoseq{\Gamma_1}{C,\Delta_1}\quad
        \twoseq{\Gamma_2,C}{\Delta_2}}{
       \twoseq{\Gamma_1,\Gamma_2}{\Delta_1,\Delta_2}}
  \hbox{\ Cut}$$
can be seen as being dual to each other \cite{girard89book}.  In
particular, the initial rule states that an occurrence of a formula on
the left is stronger than the same occurrence on the right, whereas
the cut rule states the dual: an occurrence of a formula on the
right is strong enough to remove the same occurrence from the left.
In most well designed proof systems, all occurrence of the cut-rule
can be eliminated (whether or not $C$ is an atomic formula) whereas
only non-atomic initial rules (where $C$ is non-atomic) can be
eliminated.  Atoms seem to spoil the elegant duality of the
meta-theory of these inference rules.

While the logic programming world is most comfortable with the
existence of atomic formulas, there have been a couple of recent proof
theoretic approaches that try to eliminate them entirely.  For example, in
the work on {\em definitions} and {\em fixed points} by
Schroeder-Heister \cite{schroeder-Heister93lics}, Girard
\cite{girard92mail}, and McDowell \& Miller \cite{mcdowell00tcs},  
atoms are defined to be other formulas.  In this approach, the only
primitive judgment involving terms is that of equality.  In that
setting, if definitions are {\em stratified} (no recursion through
negations) and {\em noetherian} (no infinite descent in recursion),
then all instances of cut and initial can be removed.  The setting of
{\em ludics} of Girard \cite{girard01mscs} is a more radical
presentation of logic in which atomic formulas do not exist: formulas
can be probed to arbitrary depth to uncover ``subformulas''.

Another approach to atoms is to consider {\em all} constants as
being variables.  On one hand this is a trivial position: if there are
no constants (thus, no predicate constants) there are no atomic formulas
(which are defined as formulas with non-logical constants at their
head).  On the other hand, adopting a point-of-view that constants can
vary has some appeal.  We describe this next.

\subsection{Viewing constants and variables as one}

The inference rule of $\forall$-generalization states that if $B$ is
provable then $\forall x. B$ is provable (with appropriate provisos if
the proof of $B$ depends on hypotheses).  If we are in a first-order
logic, then the free first-order variable $x$ of $B$ becomes bound in
$\forall x. B$ by this inference rule.

Observe the following two things about this rule.  First, if
we are in an untyped setting, then we can, in principle, quantify over
any variable in any expression, even those that play the role of
predicates or functions.  Mixing such rich abstractions with logic is
well known to be inconsistent so when we propose such rich
abstractions in logic, we must accompany it with some discipline (such
as typing) that will yield consistency.

Second, we need to observe that differences between constants and
variables can be seen as one of ``scope'', at least from a syntactic,
proof theoretic, and computational point of view.  For example,
variables are intended as syntactic objects that can ``vary''.  During
the computation of, say, the relation of appending lists, universal
quantified variables surrounding Horn clauses change via substitution
(via backchaining and unification) but the constructors for the empty
and non-empty lists as well as the symbol denoting the append relation
do not change and, hence, can be seen as constants.  But from a
compiling and linking point-of-view, the append predicate might be
considered something that varies: if append is in a module of Prolog
that is separately compiled, the append symbol might denote a
particular object in the compiled code that is later changed when the
code is loaded and linked.  In a similar fashion, we shall allow
ourselves to instantiate constants with expression during static
analysis.

Substituting for constants allows us to ``split the atom'': that is,
by substituting for the predicate $p$ in the atom $p(t_1,\ldots,t_n)$,
we replace that atom with a formula, which, in this paper, will be a
linear logic formula.

\subsection{Linear logic underlies computational logic}

Linear logic \cite{girard87tcs} is able to explain the proof theory of
usual Horn clause logic programming (and even richer logic programming
languages \cite{hodas94ic}).  It is also able to provide means to
reason about resources, such as items in multisets and sets.  Thus,
linear logic will allow us to sit within one declarative framework to
describe both usual logic programming as well as ``sub-atomic''
reasoning about the resources implicit in the arguments of
predicates.

\section{A primer for linear logic}\label{linear}

Linear logic connectives can be divided into the following groups: the
multiplicatives $\lpar$, $\bot$, $\ot$, $\One$; the additives
$\oplus$, $\zero$, $\with$, $\top$; the exponentials $\bang$,
$\quest$; the implications $\limp$ (where $B\limp C$ is defined as
$B^\perp\lpar C$) and $\iimp$ (where $B\iimp C$ is defined as $(\bang
B)^\perp\lpar C$); and the quantifiers $\forall$ and $\exists$
(higher-order quantification is allowed).  The equivalence of formulas
in linear logic,  $B\lequiv C$, is defined as the formula
$(B\limp C)\with(C\limp B)$.

First-order Horn clauses can be described as formulas of the form
$$\forall x_1\ldots\forall x_m [A_1\land\ldots\land A_n\supset A_0]
   \qquad (n,m\ge 0)$$
where $\land$ and $\supset$ are intuitionistic or classical logic
conjunction and implication.  There are at least two natural mappings
of Horn clauses into linear logic.  The ``multiplicative'' mapping
uses the $\ot$ and $\limp$ for the conjunction and implication: this
encoding is used in, say, the linear logic programming settings, such
as Lolli \cite{hodas94ic}, where Horn clause programming can interact
with the surrounding linear aspects of the full programming language.  
Here, we are not interested in linear logic programming per se but
with using linear logic to help establish invariants about Horn
clauses when these are interpreted in the usual, classical
setting.  As a result, we shall encode Horn clauses into linear logic
using the conjunction $\with$ and implication $\iimp$: that is, we
take Horn clauses to be formulas of the form 
$$\forall x_1\ldots\forall x_m [A_1\with\ldots\with A_n\iimp A_0].
   \qquad (n,m\ge 0)$$
The usual proof search behavior of first-order Horn clauses in
classical (and intuitionistic) logic is captured precisely when 
this style of linear logic encoding is used.

\bigbreak

\begin{figure*}
$$\begin{array}{c}
\With{\Allx{\Tok{K}}{(\Tok{append}~\Tok{nil}~\Tok{K}~\Tok{K})}}{\With{\Allx{\Tok{X}}{\Allx{\Tok{L}}{\Allx{\Tok{K}}{\Allx{\Tok{M}}{\Impliedby{(\Tok{append}~(\Tok{cons}~\Tok{X}~\Tok{L})~\Tok{K}~(\Tok{cons}~\Tok{X}~\Tok{M}))}{(\Tok{append}~\Tok{L}~\Tok{K}~\Tok{M})}}}}}}{\With{\Allx{\Tok{X}}{(\Tok{split}~\Tok{X}~\Tok{nil}~\Tok{nil}~\Tok{nil})}}{\With{\Allx{\Tok{X}}{\Allx{\Tok{A}}{\Allx{\Tok{B}}{\Allx{\Tok{R}}{\Allx{\Tok{S}}{\Impliedby{(\Tok{split}~\Tok{X}~(\Tok{cons}~\Tok{A}~\Tok{R})~(\Tok{cons}~\Tok{A}~\Tok{S})~\Tok{B})}{(\Tok{leq}~\Tok{A}~\Tok{X})\&(\Tok{split}~\Tok{X}~\Tok{R}~\Tok{S}~\Tok{B})}}}}}}}{\With{\Allx{\Tok{X}}{\Allx{\Tok{A}}{\Allx{\Tok{B}}{\Allx{\Tok{R}}{\Allx{\Tok{S}}{\Impliedby{(\Tok{split}~\Tok{X}~(\Tok{cons}~\Tok{A}~\Tok{R})~\Tok{S}~(\Tok{cons}~\Tok{A}~\Tok{B}))}{(\Tok{gr}~\Tok{A}~\Tok{X})\&(\Tok{split}~\Tok{X}~\Tok{R}~\Tok{S}~\Tok{B})}}}}}}}{\With{(\Tok{sort}~\Tok{nil}~\Tok{nil})}{\Allx{\Tok{F}}{\Allx{\Tok{R}}{\Allx{\Tok{S}}{\Allx{\Tok{Sm}}{\Allx{\Tok{B}}{\Allx{\Tok{SS}}{\Allx{\Tok{BS}}{\Impliedby{(\Tok{sort}~(\Tok{cons}~\Tok{F}~\Tok{R})~\Tok{S})}{(\Tok{split}~\Tok{F}~\Tok{R}~\Tok{Sm}~\Tok{B})\&(\Tok{sort}~\Tok{Sm}~\Tok{SS})\&(\Tok{sort}~\Tok{B}~\Tok{BS})\&(\Tok{append}~\Tok{SS}~(\Tok{cons}~\Tok{F}~\Tok{BS})~\Tok{S})}}}}}}}}}}}}}}
\end{array}
$$
\caption{Some Horn clauses for specifying a sorting relation.}
\label{one}
$$\begin{array}{c}
\With{\Allx{\Tok{K}}{\Llequiv{\Llpar{\Bottom}{\Tok{K}}}{\Tok{K}}}}{\With{\Allx{\Tok{X}}{\Allx{\Tok{L}}{\Allx{\Tok{K}}{\Allx{\Tok{M}}{\Impliedby{\Llequiv{\Llpar{\Llpar{\Llitem{\Tok{X}}}{\Tok{L}}}{\Tok{K}}}{\Llpar{\Llitem{\Tok{X}}}{\Tok{M}}}}{\Llequiv{\Llpar{\Tok{L}}{\Tok{K}}}{\Tok{M}}}}}}}}{\With{\Allx{\Tok{X}}{\Llequiv{\Llpar{\Bottom}{\Bottom}}{\Bottom}}}{\With{\Allx{\Tok{X}}{\Allx{\Tok{A}}{\Allx{\Tok{B}}{\Allx{\Tok{R}}{\Allx{\Tok{S}}{\Impliedby{\Impliedby{\Llequiv{\Llpar{\Llpar{\Llitem{\Tok{A}}}{\Tok{S}}}{\Tok{B}}}{\Llpar{\Llitem{\Tok{A}}}{\Tok{R}}}}{\One}}{\Llequiv{\Llpar{\Tok{S}}{\Tok{B}}}{\Tok{R}}}}}}}}}{\With{\Allx{\Tok{X}}{\Allx{\Tok{A}}{\Allx{\Tok{B}}{\Allx{\Tok{R}}{\Allx{\Tok{S}}{\Impliedby{\Impliedby{\Llequiv{\Llpar{\Tok{S}}{\Llpar{\Llitem{\Tok{A}}}{\Tok{B}}}}{\Llpar{\Llitem{\Tok{A}}}{\Tok{R}}}}{\One}}{\Llequiv{\Llpar{\Tok{S}}{\Tok{B}}}{\Tok{R}}}}}}}}}{\With{\Llequiv{\Bottom}{\Bottom}}{\Allx{\Tok{F}}{\Allx{\Tok{R}}{\Allx{\Tok{S}}{\Allx{\Tok{Sm}}{\Allx{\Tok{Bg}}{\Allx{\Tok{SS}}{\Allx{\Tok{BS}}{\Impliedby{\Llequiv{\Llpar{\Llitem{\Tok{F}}}{\Tok{R}}}{\Tok{S}}}
{{\Llequiv{\Llpar{\Tok{Sm}}{\Tok{B}}}{\Tok{R}}}\&
 {\Llequiv{\Tok{Sm}}{\Tok{SS}}}\&
 {\Llequiv{\Tok{B}}{\Tok{BS}}}\&
 {\Llequiv{\Llpar{\Tok{SS}}{\Llpar{\Llitem{\Tok{F}}}{\Tok{BS}}}}{\Tok{S}}}
}}}}}}}}}}}}}}
\end{array}
$$
\caption{The result of instantiating various non-logical constants in
  the above Horn clauses.}
\label{two}
\end{figure*}

\section{A primer for proof theory}\label{proof}
\label{proof theory}

A sequent is a triple of the form
$\Sigma\colon\twoseq{\Gamma}{\Delta}$ were $\Sigma$, the signature, is
a list of non-logical constants and eigenvariables paired with a
simple type, and where 
both $\Gamma$ and $\Delta$ are multisets of $\Sigma$-formulas (i.e.,
formulas all of whose non-logical symbols are in $\Sigma$).
The rules for linear logic are the standard ones \cite{girard87tcs},
except here signatures 
have been added to sequents.  The rules for quantifier introduction
are the only rules that require the signature and they are reproduced
here:
$$\frac{\Sigma,y\colon\tau;\twoseq{B[y/x],\Gamma}{\Delta}}
       {\Sigma ; \twoseq{\exists
       x^\tau.B,\Gamma}{\Delta}} \exists L \quad
  \frac{\Sigma\vdash t\colon \tau \quad \Sigma ;\twoseq{\Gamma}{B[t/x],\Delta}}
       {\Sigma;\twoseq{\Gamma}{\exists x^\tau.B,\Delta}} \exists R
$$
$$\frac{\Sigma\vdash t\colon\tau \quad \Sigma;\twoseq{B[t/x],\Gamma}{\Delta}}
       {\Sigma ;\twoseq{\forall x^\tau.B,\Gamma}{\Delta}} \forall L \quad
  \frac{\Sigma,y:\tau;\twoseq{\Gamma}{B[y/x],\Delta}}
       {\Sigma;\twoseq{\Gamma}{\forall x^\tau.B,\Delta}} \forall R
$$
The premise $\Sigma\vdash t\colon\tau$ is the judgment that
the term $t$ has the (simple) type $\tau$ given the typing declaration
contained in $\Sigma$.

We now outline three ways to instantiate things within the sequent
calculus.

\subsection{Substituting for types}
\label{subtyp}

Although we think of formulas and proofs as untyped expressions, we
shall use simple typing within sequents to control the kind of
formulas that are present.  A signature is used to bind and declare
typing for (eigen)variables and non-logical constants within a
sequent.  Simple types are, formally speaking, also a simple class of
untyped $\lambda$-terms: the type $o$ is used to denote formulas
(following Church \cite{church40}).  In a sequent calculus proof,
simple type expressions are global and admit no bindings.  As a
result, it is an easy matter to show that if one takes a proof with a
type constant $\sigma$ and replaces everywhere $\sigma$ with some
type, say, $\tau$, one gets another valid proof.  We shall do this
later when we replace a list by a multiset that approximates it: since
we are using linear logic, we shall use formulas to encode multisets
and so we shall replace the type constant {\tt list} with {\tt o}.

\subsection{Substituting for non-logical constants}
\label{subcon}

Consider the sequent
$$\Sigma,p\colon\tau ;\twoseq{\bang D_1,\bang D_2,\bang\Gamma}{p(t_1,\ldots,t_m)}$$
where the type $\tau$ is a predicate type (that is, it is of the form
$\tau_1\ra\cdots\ra\tau_m\ra o$) and where $p$ appears in, say, $D_1$
and $D_2$ and in no formula of $\Gamma$.  The linear logic exponential
$\bang$ is used here to encode the fact that the formulas $D_1$ and
$D_2$ are available for arbitrary reuse within a proof (the usual
case for program clauses).  Using the right introduction
rules for implication and the universal quantifier, it follows that
the sequent 
$$ \Sigma ; \twoseq{\bang\Gamma}{\forall p 
                            [D_1 \iimp D_2\iimp p(t_1,\ldots,t_m)]}$$
is also provable.  Since this is a universal quantifier, there must be
proofs for all instances of this quantifier.  Let $\theta$ be the
substitution $[p\mapsto\lambda x_1\ldots\lambda x_m. S]$, where $S$ is
a term over the signature $\Sigma\cup\{x_1,\ldots,x_m\}$ of type $o$.  A 
consequence of the proof theory of linear logic is that there is a
proof also of  
$$\Sigma;\twoseq{\bang\Gamma}{D_1\theta \iimp D_2\theta\iimp 
  S[t_1/x_1,\ldots,t_m/x_m]}
$$
and of the sequent
$$\Sigma;\twoseq{\bang D_1\theta,\bang D_2\theta, \bang\Gamma}
   {S[t_1/x_1,\ldots,t_m/x_m]}.$$
As this example illustrates, it is possible to instantiate a predicate
(here $p$) with an abstraction of a formula (here, $\lambda
x_1\ldots\lambda x_m.\ S$).  Such instantiation carries a provable
sequent to a provable sequent.

\subsection{Substituting for assumptions}
\label{subass}

An instance of the cut-rule (mentioned earlier) is the following: 
$$\frac{\Sigma ; \twoseq{\Gamma_1}{B}\qquad 
        \Sigma ; \twoseq{B,\Gamma_2}{C}}
       {\Sigma ; \twoseq{\Gamma_1,\Gamma_2}{C}}
$$
This inference rule (especially when associated with the
cut-elimination procedure) provides a way to merge (substitution) the
proof of a formula (here, $B$) with a use of that formula as an
assumption.   For example, consider the following situation.
Given the example in the Section~\ref{subcon}, assume that we can prove 
$$ \Sigma ; \twoseq{\bang\Gamma}{\bang D_1\theta}
   \hbox{\quad and \quad}
   \Sigma ; \twoseq{\bang\Gamma}{\bang D_2\theta}.
$$
Using two instances of the cut rule and the proofs of these sequent,
it is possible to obtain a proof of the sequent 
$$\Sigma;\twoseq{\bang\Gamma}{S[t_1/x_1,\ldots,t_m/x_m]}$$
(contraction on the left for $\bang$'ed formulas must be applied). 

Thus, by a series of instantiations of proofs, it is possible to move
from a proof of, say, 
$$\Sigma, p\colon \tau ; \twoseq{\bang D_1, \bang D_2, \bang \Gamma}
   {p(t_1,\ldots,t_m)}
$$
to a proof of 
$$\Sigma;\twoseq{\bang\Gamma}{S[t_1/x_1,\ldots,t_m/x_m]}.$$
We shall see this style of reasoning about proofs several times below.
This allows us to ``split an atom'' $p(t_1,\ldots,t_m)$ into a formula
$S[t_1/x_1,\ldots,t_m/x_m]$ and to transform proofs of the atom into
proofs of that formula.  In what follows, the formula $S$ will be a
linear logic formula that provides an encoding of some judgment about
the data structures encoded in the terms $t_1,\ldots, t_m$.

A few simple examples of using higher-order instantiations of logic
programs in order to help reasoning about them appear in
\cite{miller02amast}.

\section{Encoding multisets as formulas}\label{msets}

We wish to encode multisets and sets and simple judgments about them
(such as inclusion and equality) as linear logic formulas.  We
consider multisets first.   Let token
{\sl item} be a linear logic predicate of one argument: the 
linear logic atomic formula $\llitem{x}$ will denote the multiset
containing just the one element $x$ occurring once.
There are two natural encoding of multisets into formulas using this
predicate.  The {\em conjunctive} encoding uses $\One$ for the empty
multiset and $\ot$ to combine two multisets.  For example, the
multiset $\{ 1,2,2\}$ is encoded by the linear logic formula
$\llitem{1}\ot\llitem{2}\ot\llitem{2}$.  Proofs search using this
style encoding places multiset on the left of the sequent arrow.  This
approach is favored when an intuitionistic subset of linear logic
is used, such as in Lolli \cite{hodas94ic}, LinearLF
\cite{cervesato96lics}, and MSR \cite{cervesato99csfw}.  The dual
encoding, the {\em disjunctive} encoding, uses $\bot$ for the empty
multiset and $\lpar$ to combine two multisets.  Proofs search using
this style encoding places multisets on the right of the sequent
arrow.  Multiple conclusion sequents are now required.  Systems such
as LO \cite{andreoli91ngc} and Forum \cite{miller96tcs} use this style
of encoding.  If negation is available, then the choice of which
encoding one chooses is mostly a matter of style.  We pick the
disjunctive encoding for the rather shallow reason that the inclusion
judgment for multisets and sets is encoded as an implication instead
of a reverse implication, as we shall now see.

Let $S$ and $T$ be the two formulas $\llitem
s_1\lpar\cdots\lpar\llitem s_n$ and $\llitem
t_1\lpar\cdots\lpar\llitem t_m$, respectively ($n,m\ge 0$).  Notice
that $\vdash S\limp T$ if and only if $\vdash T\limp S$ if and only if
the two multisets $\{s_1,\ldots, s_n\}$ and $\{t_1,\ldots, t_m\}$ are
equal.  Consider now, however, the following two ways for encoding the
multiset inclusion $S\mssub T$.
\begin{itemize}
\item $S\lpar 0\limp T$.  This formula mixes multiplicative
  connectives with the additive connective $0$: the latter allows
  items that are not matched between $S$ and $T$ to be deleted.

\item $\exists q (S\lpar q\limp T)$.  This formula mixes
  multiplicative connectives with a higher-order quantifier.  While we
  can consider the instantiation for $q$ to be the multiset difference
  of $S$ from $T$, there is no easy way in the logic to enforce that
  interpretation of the quantifier.
\end{itemize}
As it turns out, these two approaches are equivalent in linear logic:
in particular, $\vdash\zero\lequiv\forall p. p$ (linear logic absurdity) and 
$$\vdash\forall S\forall T [(S\lpar 0\limp T)\lequiv\exists q (S\lpar q\limp T)].$$
Thus, below we can choose either one of these encodings for 
multiset inclusion.

\section{Multisets approximations}\label{multisets}

A {\em multiset expression} is a formula in linear logic built from
the predicate symbol {\sl item} (denoting the singleton multiset),
the linear logic multiplicative disjunction $\lpar$ (for multiset
union), and the unit $\bot$ for $\lpar$ (used to denote the empty
multiset).  We shall also allow a predicate variable (a variable of
type $o$) to be used to denote a (necessarily open) multiset
expression.  An example of an open multiset expression is
$\llitem{f(X)}\lpar\bot\lpar Y$, where $Y$ is a variable of type $o$,
$X$ is a first-order variable, and $f$ is some first-order term
constructor.

Let $S$ and $T$ be two multiset expressions.  The two {\em multiset
judgments} that we wish to capture are multiset inclusion, written as
$S\mssub T$, and equality, written as $S\mseq T$.  We shall use the
syntactic variable $\rho$ to range over these two judgments, which are
formally binary relations of type $o\ra o\ra o$.
A {\em multiset statement} is a formula of the form
$$\forall\bar x[S_1\Rho_1 T_1\with\cdots\with S_n\Rho_n T_n 
           \iimp S_0\Rho_0 T_0 ]$$
where the quantified variables $\bar x$ are either first-order or of
type $o$ and formulas $S_0,T_0,\ldots,S_n,T_n$ are possibly open
multiset expressions.  

If $S$ and $T$ are closed multiset expressions, then we
write $\msmodels S\mssub T$ whenever the multiset (of closed
first-order terms) denoted by $S$ is contained in the multiset denoted
by $T$, and we write $\msmodels S\mseq T$ whenever the multisets
denoted by $S$ and $T$ are equal.   Similarly, we write
$$\msmodels\forall\bar x[S_1\Rho_1 T_1\with\cdots\with 
                         S_n\Rho_n T_n\iimp S_0\Rho_0 T_0 ]
$$ 
if for all closed substitutions $\theta$ such that $\msmodels
S_i\theta\Rho_i T_i\theta$ for all $i=1,\ldots,n$, it is the case that
$\msmodels S_0\theta\Rho_0 T_0\theta$.

The following Proposition is central to our use of linear logic to
establish multiset statements for Horn clause programs.

\begin{prop}\label{ms ok}
Let $S_0, T_0, \ldots, S_n, T_n$ ($n\ge0$) be multiset expressions all
of whose free variables are in the list of variables $\bar x$.  For
each judgment $s\Rho t$ we write $s\hatRho t$ to denote $\exists
q(s\lpar q\limp t)$ if $\Rho$ is $\mssub$ and $t \lequiv s$ if $\Rho$
is $\mseq$.  If 
$$\forall \bar x [
   S_1\hatRho_1 T_1 \with\ldots\with
   S_n\hatRho_n T_n \iimp S_0\hatRho_0 T_0 ]
$$
is provable in linear logic, then 
$$\models_{ms}
  \forall\bar x[S_1\Rho_1 T_1\with\cdots\with S_n\Rho_n T_n 
                \iimp S_0\Rho_0 T_0 ]$$
\end{prop}

This Proposition shows that linear logic can be used in a sound way to
infer valid multiset statement.  On the other hand,  the converse
(completeness) does not hold: the statement 
$$\forall x\forall y. (x\mssub y)\with (y\mssub x)\iimp (x \mseq y)$$
is valid but its translation into linear logic is not provable.

\begin{figure*}
$$\begin{array}{c}
\Allx{\Tok{X}}{(\Tok{split}~\Tok{X}~\Tok{nil}~\Tok{nil}~\Tok{nil})}\\
\Allx{\Tok{X}}{\Allx{\Tok{B}}{\Allx{\Tok{R}}{\Allx{\Tok{S}}{\Impliedby{(\Tok{split}~\Tok{X}~(\Tok{cons}~\Tok{X}~\Tok{R})~\Tok{S}~\Tok{B})}{(\Tok{split}~\Tok{X}~\Tok{R}~\Tok{S}~\Tok{B})}}}}}
\\
\Allx{\Tok{X}}{\Allx{\Tok{A}}{\Allx{\Tok{B}}{\Allx{\Tok{R}}{\Allx{\Tok{S}}{\Impliedby{(\Tok{split}~\Tok{X}~(\Tok{cons}~\Tok{A}~\Tok{R})~(\Tok{cons}~\Tok{A}~\Tok{S})~\Tok{B})}{(\Tok{lt}~\Tok{A}~\Tok{X})\&(\Tok{split}~\Tok{X}~\Tok{R}~\Tok{S}~\Tok{B})}}}}}}
\\
\Allx{\Tok{X}}{\Allx{\Tok{A}}{\Allx{\Tok{B}}{\Allx{\Tok{R}}{\Allx{\Tok{S}}{\Impliedby{(\Tok{split}~\Tok{X}~(\Tok{cons}~\Tok{A}~\Tok{R})~\Tok{S}~(\Tok{cons}~\Tok{A}~\Tok{B}))}{(\Tok{gr}~\Tok{A}~\Tok{X})\&(\Tok{split}~\Tok{X}~\Tok{R}~\Tok{S}~\Tok{B})}}}}}}
\end{array}
$$
\caption{A change in the specification of splitting lists to drop
         duplicates.} 
\label{three}
\newcommand{\Splitinv}[4]{(\quest #2\limp\quest(\Llitem{#1}\oplus#3\oplus#4))}
$$\begin{array}{c}
\Allx{\Tok{X}}{\Splitinv{\Tok{X}}{\zero}{\zero}{\zero}}\\
\Allx{\Tok{X}}{\Allx{\Tok{B}}{\Allx{\Tok{R}}{\Allx{\Tok{S}}{
\Impliedby
   {\Splitinv{\Tok{X}}{(\Llitem{\Tok{X}}\oplus\Tok{R})}{\Tok{S}}{\Tok{B}}}
   {\Splitinv{\Tok{X}}{\Tok{R}}{\Tok{S}}{\Tok{B}}}
  }}}}\\
\Allx{\Tok{X}}{\Allx{\Tok{A}}{\Allx{\Tok{B}}{\Allx{\Tok{R}}{\Allx{\Tok{S}}{
   \Impliedby
            {\Splitinv{\Tok{X}}{(\Llitem{\Tok{A}}\oplus\Tok{R})}{\Llitem{\Tok{A}}\oplus\Tok{S}}{\Tok{B}}}
   { \One\&\Splitinv{\Tok{X}}{\Tok{R}}{\Tok{S}}{\Tok{B}}}
}}}}}
\\
\Allx{\Tok{X}}{\Allx{\Tok{A}}{\Allx{\Tok{B}}{\Allx{\Tok{R}}{\Allx{\Tok{S}}{
   \Impliedby
            {\Splitinv{\Tok{X}}{(\Llitem{\Tok{A}}\oplus\Tok{R})}{\Tok{S}}{\Llitem{\Tok{A}}\oplus\Tok{B}}}
   { \One\&\Splitinv{\Tok{X}}{\Tok{R}}{\Tok{S}}{\Tok{B}}}
}}}}}
\end{array}
$$
\caption{The result of substituting set approximations into the {\tt
    split} program.}
\end{figure*}

To illustrate how deduction in linear logic can be used to establish
the validity of a multiset statement, consider the first-order Horn
clause program in Figure~\ref{one}.  The signature for this collection
of clauses can be given as follows:
\begin{verbatim}
nil    : list
cons   : int -> list -> list
append : list -> list -> list -> o
split  : int -> list -> list -> list -> o
sort   : list -> list -> o
leq    : int -> int -> o
gr     : int -> int -> o
\end{verbatim}
The first two declarations provide constructors for empty and non-empty
lists, the next three are predicates whose  Horn clause definition
is presented in Figure~\ref{one}, and the last two are order relations
that are apparently defined elsewhere.  

If we think of lists as collections of items, then we might want to
check that the sort program as written does not drop, duplicate, or
create any elements.
That is, if the atom $(\Tok{sort}~s~t)$ is provable then the multiset
of items in the list denoted by $s$ is equal to the multiset of items
in the list denoted by $t$.  If this property holds then $t$ and $s$
are lists that are permutations of each other: of course, this does
not say that it is the correct permutation but this more simple fact
is one that, as we show, can be inferred automatically.

Computing this property of our example logic programming follows the
following three steps.

First, we provide an approximation of lists as being, in fact,
multiset: more precisely, as {\em formulas} denoting multisets.  The
first step, therefore, must be to substitute \Tok{o} for \Tok{list} in
the signature above.  Now we can now interpret the constructors for
lists using the substitution
$$\Tok{nil}\mapsto \bot \qquad \Tok{cons}\mapsto \lambda x\lambda
y.\ \Llitem{x}\lpar y.$$ 
Under such a mapping, the list (\Tok{cons}~{\Tok 1}~(\Tok{cons}~{\Tok
3}~(\Tok{cons}~{\Tok 2}~\Tok{nil}))) is mapped to the multiset
expression $\Llitem{1}\lpar\Llitem{3}\lpar\Llitem{2}\lpar\bot$.

Second, we associate with each predicate in Figure~\ref{one} a
multiset judgment that encodes an invariant concerning the multisets
denoted by the predicate's arguments.  For example, if
$(\Tok{append}~r~s~t)$ or $(\Tok{split}~u~t~r~s)$ is provable then the
multiset union of the items in $r$ with those in $s$ is equal to the
multiset of items in $t$, and  if $(\Tok{sort}~s~t)$ is provable then the
multisets of items in lists $s$ and $t$ are equal.  This association
of multiset judgments to atomic formulas can be achieved formally using
the following substitutions for constants:
$$\halign{\qquad\hfill $# $ & $\mapsto #$\hfill\cr
  \Tok{append} & \lambda x\lambda y\lambda z.\ (x\lpar y)\lequiv z\cr
   \Tok{split} & \lambda u\lambda x\lambda y\lambda z.\ (y\lpar z)\lequiv x\cr
   \Tok{sort}  & \lambda x\lambda y.\ x\lequiv y\cr
}$$
The predicates \Tok{leq} and \Tok{gr} (for the least-than-or-equal-to
and greater-than relations) make no statement about collections of
items, so that they can be mapped to a trivial tautology via the
  substitution 
$$\Tok{leq}\mapsto\lambda x\lambda y.\ \One \qquad
  \Tok{gr }\mapsto\lambda x\lambda y.\ \One
$$
Figure~\ref{two} presents the result of applying these mappings to
Figure~\ref{one}.

Third, we must now attempt to prove each of the resulting formulas.
In the case of Figure~\ref{two}, all the displayed formulas are
trivial theorems of linear logic.

Having taken these three steps, we now claim that we have proved the
intended collection judgments associate to each of the logic
programming predicates above: in particular, we have now shown that
our particular sort program computes a permutation.

\section{Formalizing the method}\label{formalizing}

The formal correctness of this three stage approach is easily
justified given the substitution properties we presented in
Section~\ref{proof theory} for the sequent calculus presentation of
linear logic.

Let $\Gamma$ denote a set of formulas that contains those in
Figure~\ref{one}.  Let $\theta$ denote the substitution described
above for the type {\tt list}, for the constructors {\tt nil} and {\tt
cons}, and for the predicates in Figure~\ref{one}.  If $\Sigma$ is the
signature for $\Gamma$ then split $\Sigma$ into the two signatures
$\Sigma_1$ and $\Sigma_2$ so that $\Sigma_1$ is the domain of the
substitution $\theta$ and let $\Sigma_3$ be the signature of the range
of $\theta$ (in this case, it just contains the constant
{\sl item}).
Thus, $\Gamma\theta$ is the set of formula in Figure~\ref{two}.  

Assume now that $\Sigma_1,\Sigma_2;\twoseq{\Gamma}{sort(t,s)}$ is
provable.  Given the discussion in Sections~\ref{subtyp}
and~\ref{subcon}, we know that
$$\Sigma_1,\Sigma_3;\twoseq{\Gamma\theta}{t\theta\lequiv s\theta}$$ is
provable.  Since the formulas in $\Gamma\theta$ are provable, we can
use substitution into proofs (Section~\ref{subass}) to conclude that
$\Sigma_1,\Sigma_3;\twoseq{}{t\theta\lequiv s\theta}$.  Given
Proposition~\ref{ms ok}, we can conclude that $\msmodels t\theta\mseq
s\theta$: that is, that $t\theta$ and $s\theta$
encode the same multiset.

Consider the following model theoretic argument for establishing
similar properties of Horn clauses.  Let $\Mscr$ be
the Herbrand model
that captures the invariants that we have in mind.  In
particular, $\Mscr$ contains the atoms $({\tt append}~r~s~t)$
and $({\tt split}~u~t~r~s)$ if the items in the list $r$ added to the
items in list $s$ are the same as the items in $t$.  Furthermore,
$\Mscr$ contains all closed atoms of the form $({\tt leq}~t~s)$ and
$({\tt gr}~t~s)$, and closed atoms $({\tt sort}~s~t)$ where $s$ and $t$
are lists that are permutations of one another.  One can now show
that $\Mscr$ satisfies all the Horn clauses in Figure~\ref{one}.  As a
consequence of the soundness of first-order classical logic, any atom
provable from the clauses in Figure~\ref{one}, must be true in
$\Mscr$.  By construction of $\Mscr$, this means that the desired 
invariant holds for all atoms proved from the program.

The approach suggested here using linear logic and deduction remains
syntactic and proof theoretic: in particular, showing that a model
satisfies a Horn clause is replaced by a deduction within linear
logic.

\section{Sets approximations}\label{sets}

It is rather easy to encode sets and the equality and subset
judgments on sets into linear logic.  In fact, the transition to set
from multiset is provided by the use of the linear logic exponential:
since we are using disjunctive encoding of collections (see the
discussion in Section~\ref{msets}), we use the $\quest$ exponential
(if we were using the conjunctive encoding, we would use the $\bang$
exponential).

The expression $\quest\llitem{t}$ can be seen as describing the
presence of an item for which the exact multiplicity does not matter:
this formula represents the capacity to be used any number of times.
Thus, the set $\{x_1,\ldots,n_n\}$ can be encoded as
$\quest\llitem{x_1}\lpar\cdots\lpar\quest\llitem{x_n}$.  Using logical
equivalences of linear logic, this formula is also equivalent to the
formula $\quest(\llitem{x_1}\oplus\cdots\oplus\llitem{x_n})$.  This
latter encoding is the one that we shall use for building our encoding
of sets.

A {\em set expression} is a formula in linear logic
built from the predicate symbol {\sl item} (denoting the the singleton
set), the linear logic additive disjunction $\oplus$ (for set union),
and the unit $\zero$ for $\oplus$ (used to denote the empty set).  We
shall also allow a predicate variable (a variable of type $o$) to be
used to denote a (necessarily open) set expression.  An example of an
open multiset expression is $\llitem{f(X)}\oplus\zero\oplus Y$, where
$Y$ is a variable of type $o$, $X$ is a first-order variable, and $f$
is some first-order term constructor.

\newcommand{\setsub}{\subseteq}
\newcommand{\seteq }{\mathrel{\buildrel s\over =}}
\newcommand{\setmodels }{\models_{s}}

Let $S$ and $T$ be two set expressions.  The two {\em set
judgments} that we wish to capture are set inclusion, written as
$S\setsub T$, and equality, written as $S\seteq T$.  We shall use the
syntactic variable $\rho$ to range over these two judgments, which are
formally binary relations of type $o\ra o\ra o$.
A {\em set statement} is a formula of the form
$$\forall\bar x[S_1\Rho_1 T_1\with\cdots\with S_n\Rho_n T_n 
           \iimp S_0\Rho_0 T_0 ]$$
where the quantified variables $\bar x$ are either first-order or of
type $o$ and formulas $T_0,S_0,\ldots,T_n,S_n$ are possibly open
set expressions.  

If $S$ and $T$ are closed set expressions, then we
write $\setmodels S\setsub T$ whenever the set (of closed
first-order terms) denoted by $S$ is contained in the set denoted
by $T$, and we write $\setmodels S\seteq T$ whenever the sets
denoted by $S$ and $T$ are equal.   Similarly, we write
$$\setmodels\forall\bar x[S_1\Rho_1 T_1\with\cdots\with 
                         S_n\Rho_n T_n\iimp S_0\Rho_0 T_0 ]
$$ 
if for all closed substitutions $\theta$ such that $\setmodels
S_i\theta\Rho_i T_i\theta$ for all $i=1,\ldots,n$, it is the case that
$\setmodels S_0\theta\Rho_0 T_0\theta$.

The following Proposition is central to our use of linear logic to
establish set statements for Horn clause programs.

\begin{prop}\label{set ok}
Let $S_0, T_0, \ldots, S_n, T_n$ ($n\ge0$) be set expressions all
of whose free variables are in the list of variables $\bar x$.  For
each judgment $s\Rho t$ we write $s\hatRho t$ to denote $\quest s\limp
\quest t$ if $\Rho$ is $\setsub$ and $(\quest s\limp \quest
t)\with(\quest t\limp \quest s)$ if $\Rho$ is $\seteq$.  If 
$$\forall \bar x [
   S_1\hatRho_1 T_1 \with\ldots\with
   S_n\hatRho_n T_n \iimp S_0\hatRho_0 T_0 ]
$$
is provable in linear logic, then 
$$\models_{s}
  \forall\bar x[S_1\Rho_1 T_1\with\cdots\with S_n\Rho_n T_n 
                \iimp S_0\Rho_0 T_0 ]$$
\end{prop}

Lists can be approximated by sets by using the following substitution: 
$$\Tok{nil}\mapsto \zero \qquad \Tok{cons}\mapsto \lambda x\lambda
y.\ \Llitem{x}\oplus y.$$ 
Under such a mapping, the list (\Tok{cons}~{\Tok 1}~(\Tok{cons}~{\Tok
2}~(\Tok{cons}~{\Tok 2}~\Tok{nil}))) is mapped to the set
expression $\Llitem{1}\oplus\Llitem{2}\oplus\Llitem{2}\oplus\zero$.
This expression is equivalent ($\lequiv$) to the set expression 
$\Llitem{1}\oplus\Llitem{2}$.

For a simple example of using set approximates, consider modifying the
sorting program provided before so that duplicates are not kept in the
sorted list.  Do this modification by replacing the previous
definition for splitting a list with the clauses in
Figure~\ref{three}.  That figure contains a new definition of
splitting that contains three clauses for deciding
whether or not the ``pivot'' for the 
splitting {\tt X} is equal to, less than (using the {\tt lt}
predicate), or greater than the first 
member of the list being split.  Using the following substitutions for
predicates 
$$\halign{\qquad\hfill $# $ & $\mapsto #$\hfill\cr
 \Tok{append} & \lambda x\lambda y\lambda z.\ \quest (x\oplus y)\lequiv\quest z\cr
  \Tok{split} & \lambda u\lambda x\lambda y\lambda z.\ 
                  \quest x \limp \quest(\llitem{u}\oplus y\oplus z)\cr
  \Tok{sort}  & \lambda x\lambda y.\ \quest x\lequiv \quest y\cr
}$$
(as well as the trivial substitution for \Tok{lt} and \Tok{ge}), 
we can show that sort relates two lists only if those lists are
approximated by the same set.

\begin{figure}
$$  \frac{}{\twoseq{\Gamma ; A_i}{A_1\oplus\cdots\oplus A_n}}
  \hbox{\ $\oplus$ R}
$$
$$
\frac{\twoseq{\Gamma ; A_1}{C} \quad \ldots\quad
        \twoseq{\Gamma ; A_n}{C}}
       {\twoseq{\Gamma ; A_1\oplus\cdots\oplus A_n}{C}}
  \hbox{\ $\oplus$ L}
$$
$$
  \frac{\twoseq{\Gamma ; B_1\oplus\cdots\oplus B_m}{C}} 
       {\twoseq{\Gamma ; A}{C}}
  \hbox{\ BC}
$$
Here, $n,m\ge 0$ and in the BC (backchaining) inference rule, 
the formula
$
 \quest(A_1\oplus\cdots\oplus A_n) \limp \quest(B_1\oplus\cdots\oplus
 B_m) 
$ must be a member of $\Gamma$ and $A\in\{A_1,\ldots,A_n\}$.
\caption{Specialized proof rules for proving set statements.}
\label{additive}
\end{figure}

In the case of determining the validity of a set statement, the use of
linear logic here appears to be rather weak when compared to the large
body of results for solving set-based constraint systems
\cite{aiken94ppcp,pacholski97cp}. 

\section{Automation of deduction}\label{automation}

We describe how automation of proof for the linear logic translations
of set and multiset statements given in Propositions~\ref{ms ok}
and~\ref{set ok} can be performed. 

In order to understand how to automatically prove the required
formulas, we first provide a normal form theorem for the
fragment of linear logic for which we are interested.  The key result
of linear logic surrounding the search for cut-free proofs is given by
the completeness of {\em focused proofs} \cite{andreoli92jlc}.
Focused proofs are a normal form that significantly generalizes
standard completeness results in logic programming, including the
completeness of SLD-resolution and uniform proofs as well as various
forms of bottom-up and top-down reasoning.

We first analyze the nature of proof search for the linear logic
translation of set statements.  Note that when considering provability
of set statements, there is no loss of generality if the only set
judgment it contains is the subset judgment since set equality can be
expressed as two inclusions.  We now prove that the proof system in
Figure~\ref{additive} is sound and complete for proving set
statements. 

\begin{prop}\label{additive ps}
Let $S_0, T_0, \ldots, S_n, T_n$ ($n\ge0$) be set expressions all
of whose free variables are in the list of variables $\bar x$.  The
formula 
$$\forall \bar x
  [(\quest S_1\limp \quest T_1)\with\ldots\with
   (\quest S_n\limp \quest T_n)\iimp (\quest S_0\limp \quest T_0)]
$$
is provable in linear logic if and only if the sequent 
$$\twoseq{(\quest S_1\limp \quest T_1),\ldots,(\quest S_n\limp
  \quest T_n); S_0}{T_0}$$
is provable using the proof system in Figure~\ref{additive}. 
\end{prop}

\begin{proof}
The soundness part of this proposition (``if'') is easy to show.  For
completeness (``only if''), we use the completeness of focused proofs
in \cite{andreoli92jlc}.  In order to use this result of focused
proofs, we need to give a polarity to all atomic formulas.  We do this
by assigning all atomic formulas (those of the form $\llitem{(\cdot)}$
and those symbols in $\bar x$ of type $o$) negative polarity.  Second,
we need to translation the two sided sequent $\twoseq{\Gamma;S}{T}$ to
$\Gamma^\perp; T\Uparrow S^\perp$ when $S$ is not atomic (that is, its
top-level logical connective is $\oplus$) and to $\Gamma^\perp,
T;S^\perp\Uparrow\cdot$ when $S$ is a atom.  Completeness then follows
directly from the structure of focused proofs.
\end{proof}

Notice that the resulting proofs are essentially bottom-up: one
reasons from formulas on the left of the sequent arrow to formulas on
the right.

We can now conclude that it is decidable to determine whether or not
the linear logic translation of a set statement is provable. 
Notice that in a proof built using the inference rules in
Figure~\ref{additive}, if the endsequent is
$\twoseq{\Gamma;S}{T}$ then all sequents in the proof have
the form $\twoseq{\Gamma;S'}{T}$, for some $S'$.  Thus, the search for
a proof either succeeds (proof search ends by placing $\oplus$ R on
top), or fails to find a proof, or it cycles, a case we can always
detect since there is only a 
finite number of atomic formulas that can be $S'$.

\begin{figure}
$$\frac{}
  {\twoseq{\Gamma;A_1\lpar\cdots\lpar A_n}{A_1,\ldots,A_n}}
\hbox{\ $\lpar$ L}
$$
$$
\frac{\twoseq{\Gamma ; S}{T_1, T_2,\Delta}}
     {\twoseq{\Gamma ; S}{T_1\lpar T_2,\Delta}}
  \hbox{\ $\lpar$ R}
$$
$$
  \frac{\twoseq{\Gamma ; S}{A_1,\ldots, A_n,\Delta}}
       {\twoseq{\Gamma ; S}{B_1,\ldots, B_m,\Delta}}
  \hbox{\ BC}
$$
Here, $n,m\ge 0$ and in the BC (backchaining) inference rule, 
it must be the case that the formula 
$$(A_1\lpar\cdots\lpar A_n) \limp (B_1\lpar\cdots\lpar B_m) 
$$
is a member of $\Gamma$.

\caption{Specialized proof rules for proving multiset statements.}
\label{multiplicative}
\end{figure}

The proof system in Figure~\ref{multiplicative} can be used to
characterize the structure of proofs of the linear logic encoding of
multiset statements.  Let 
$$\forall \bar x [
   S_1\hatRho_1 T_1 \with\ldots\with
   S_n\hatRho_n T_n \iimp S_0\hatRho_0 T_0 ]
$$
be the translation of a multiset statement into linear logic.
Provability of this formula can be reduced to attempting to prove 
$S_0\hatRho_0 T_0$ from assumptions of the form 
$$(A_1\lpar\cdots\lpar A_n) \limp (B_1\lpar\cdots\lpar B_m),$$
where $A_1,\ldots,A_n,B_1,\ldots,B_m$ are atomic formulas.  Such
formulas can be called {\em multiset rewriting clauses} since
back\-chaining on such clauses amounts to rewriting the right-hand-side
multiset of a sequent (see rule BC in
Figure~\ref{multiplicative}).  Such rewriting clauses are particularly
simple since they do not involve quantification.

\begin{prop}\label{multiplicative ps}
Let $S_0$ and $T_0$ be multiset expressions all of whose free
variables are in the list of variables $\bar x$ and let $\Gamma$ be a
set of multiset rewriting rules.  The formula $S_0\limp T_0$ is a
linear logic consequence of $\Gamma$ if and only if the sequent 
$\twoseq{\Gamma ; S_0}{T_0}$ is provable using the
inference rules in Figure~\ref{multiplicative}.
\end{prop}

\begin{proof}
The soundness part of this proposition (``if'') is easy to show.
Completeness (``only if'') is proved elsewhere, for example, in
\cite[Proposition 2]{miller92welp}.  It is also an easy consequence of
the the completeness of focused proofs in \cite{andreoli92jlc}: fix
the polarity to all atomic formulas to be positive.
\end{proof}

Notice that the proofs using the rules in Figure~\ref{multiplicative}
are straight line proofs (no branching) and that they are top-down (or
goal-directed).  Given these observation, it follows that determining
if $S_0\limp T_0$ is provable from a set of multiset rewriting clauses
is decidable, since this problem is contained within the reachability
problem of Petri Nets \cite{esparza94eatcs}.  Proving a multiset
inclusion judgment $\exists q (S_0\lpar q\limp T_0)$ involves first
instantiating this higher-order quantifier.  In principle, this
instantiation can be delayed until attempting to apply the sole
instance of the $\lpar$ L rule (Figure~\ref{multiplicative}).

\section{List approximations}\label{lists}

We now consider using lists as approximations.  Since lists have more
structure than sets and multisets, it is more involved to encode and
reason with them.  We only illustrate their use and do not follow a
full formal treatment for them.

Since the order of elements in a list is important, the encoding of
lists into linear logic must involve a connective that is not
commutative.  (Notice that both $\lpar$ and $\oplus$ are commutative.)
Linear implication provides a good candidate for encoding the order
used in lists.
For example, consider proof search with the formula 
$$\llitem{a}\bla(p\bla (\llitem{b}\bla (p\bla \bot)))$$ on the
right.  (This formula is equivalent to $\llitem{a}\lpar (p^\bot\ot
(\llitem{b}\lpar p^\bot))$.)  Such a formula can be seen as describing
a process that is willing to output the item $a$ then go into input
mode waiting for the atomic formula $p$ to appear.  If that formula
appears, then item $b$ is output and again it goes into input waiting
mode looking for $p$.  If another occurrence of $p$ appears, this
process becomes the inactive process.  Clearly, $a$ is output
prior to when $b$ is output: this ordering is faithfully captured by
proof search in linear logic.  Such an encoding of asynchronous process
calculi into linear logic has been explored in a number of papers:
see, for example, \cite{kobayashi94fac,miller03fcs}.

The example above suggests that lists and list equality can be
captured directly in linear logic using the following encoding:
$$ \Tok{nil} \mapsto \lambda l.\bot \qquad
   \Tok{cons}\mapsto \lambda x\lambda R\lambda l.\ 
                     \llitem{x}\bla (l\bla (R~l))
$$
The encoding of the list, say
$(\Tok{cons}~a~(\Tok{cons}~b~\Tok{nil}))$,
is given by the $\lambda$-abstraction 
$$\lambda l. \llitem{a}\bla(l\bla (\llitem{b}\bla (l\bla \bot))).$$

The following proposition can be proved by induction on the length of
the list $t$. 

\begin{prop}\label{list eq}
Let $s$ and $t$ be two lists (built using \Tok{nil} and \Tok{cons})
and let $S$ and $T$ be the translation of those lists into expressions
of type $o\ra o$ via the substitution above.  Then $\forall l. (S l)
\lequiv (T l)$ is provable in 
linear logic if and only if $s$ and $t$ are the same list.
\end{prop}

This presentation of lists can be ``degraded'' to multisets simply by
applying the translation of a list to the formula $\bot$.  For
example, applying the translation of
$(\Tok{cons}~a~(\Tok{cons}~b~\Tok{nil}))$ to $\bot$ yields the
formulas $$\llitem{a}\bla(\bot\bla (\llitem{b}\bla (\bot\bla \bot)))$$
which is linear logically equivalent to $\llitem{a}\lpar \llitem{b}$.

Given this presentation of lists, there appears to be no simple
combinator for, say, list concatenation and, as a result, there is no
direct way to express the judgments of prefix, suffix, sublist, etc.
Thus, beyond equality of lists (by virtual of Proposition~\ref{list
eq}) there are few natural judgments that can be stated for list.
More can be done, however, by considering difference lists.

\section{Difference list approximations}\label{dlists}

Since our framework includes $\lambda$-abstractions, it is natural to
represent difference lists as a particular kind of list abstraction
over a list.  For example, in $\lambda$Prolog a difference list is
naturally represented as a $\lambda$-term of the form
$$ \lambda L. \Tok{cons}~x_1~(\Tok{cons}~x_2~(\ldots
   (\Tok{cons}~x_n~L)\ldots)).
$$ 
Such abstracted lists are appealing since the simple
operation of composition encodes the concatenation of two lists.
Given concatenation, it is then easy to encode the judgments of prefix
and suffix.  To see other example of computing on difference lists
described in fashion, see \cite{brisset91iclp}.

\begin{figure*}
$$\begin{array}{c}
{\tt~(traverse~emp~null)}\\
\Allx{\Tok{N}}{\Allx{\Tok{R}}{\Allx{\Tok{S}}{
 {\tt~(traverse~R~S)}\iimp {\tt~(traverse~(bt~N~emp~R)~(cons~N~S))}}}}\\
\Allx{\Tok{N}}{
\Allx{\Tok{M}}{
\Allx{\Tok{R}}{
\Allx{\Tok{S}}{
\Allx{\Tok{L1}}{
\Allx{\Tok{L2}}{
  {\tt~(traverse~(bt~M~L1~(bt~N~L2~R))~S)}
 \iimp
  {\tt~(traverse~(bt~N~(bt~M~L1~L2)~R)~S)}}}}}}}
\end{array}
$$
\caption{Traversing a binary tree to produce a list.}
\label{traverse}
$$\begin{array}{c}
\Allx{W}{\Allx{w}{W w \lequiv W w}}\\
 \Allx{N}{\Allx{R}{\Allx{S}{\Allx{W}{\Allx{w}{
    \llitem{N} \bla (w \bla R~W~w) \lequiv (\llitem{N} \bla (w \bla S~W~w)) \bla 
    \Allx{W}{\Allx{w}{R~W~w \lequiv S~W~w}}}}}}}\\
 \Allx{N}{\Allx{M}{\Allx{L_1}{\Allx{L_2}{\Allx{R}{\Allx{S}{\Allx{W}{\Allx{w}{\\
\qquad\qquad\quad   
  L_1 (\lambda k. item M \bla (k \bla L_2 (\lambda l. item N \bla (l \bla R~W~l)) k)) w \lequiv S~W~w \bla \\
\Allx{W}{\Allx{w}{L_1 (\lambda k. item M \bla (k \bla L_2 (\lambda l. item N \bla (l \bla R~W~l)) k)) w \lequiv S~W~w}}}}}}}}}}
\end{array}
$$
\caption{Linear logic formulas arising from a difference list approximation.}
\label{traverse ll}
\end{figure*}

Lists can be encoded using the difference list notion with the
following mapping into linear logic formulas.
$$\halign{\qquad\hfill $# $ & $\mapsto #$\hfill\cr
  \Tok{nil}  & \lambda L\lambda l.~L~l \cr
  \Tok{cons} & \lambda x\lambda R\lambda L\lambda l.\ 
                     \llitem{x}\bla (l\bla (R~L~l))\cr
}$$
The encoding of the list, say
$(\Tok{cons}~a~(\Tok{cons}~b~\Tok{nil}))$,
is given by the $\lambda$-abstraction 
$$\lambda L\lambda l. \llitem{a}\bla(l\bla 
                     (\llitem{b}\bla (l\bla L~l))).$$

In Figure~\ref{traverse}, a predicate  for traversing a
binary tree is given.  Binary trees are encoded using the type
\Tok{btree} and are constructed using the constructors \Tok{emp}, for
the empty tree, and \Tok{bt} of type
$\Tok{int}\ra\Tok{btree}\ra\Tok{btree}\ra\Tok{btree}$, for building
non-empty trees.  A useful invariant of this program is that the list
of items approximating the binary tree structure in the first
argument of \Tok{traverse} is equal to the list of items in the second
argument.   Linear logic formulas for computing that approximation can
be generated using the following approximating substitution.
$$\halign{\hfill $# $ & $\mapsto #$\hfill\cr
  \Tok{btree}& \Tok{o}\cr
  \Tok{emp}  & \lambda L\lambda l.~L~l \cr
  \Tok{bt}   & \lambda x\lambda R\lambda S\lambda L\lambda l.
               (R~(\lambda l. \llitem{x} \bla (l \bla (S~L~l)))~l)) \cr
}$$
The result of applying that substitution (as well as the one above for
\Tok{nil} and \Tok{cons}) is displayed in Figure~\ref{traverse ll}.
While these formulas appear rather complex, they are all, rather
simple theorems of higher-order linear logic: these 
theorems are essentially trivial since the 
$\lambda$-conversions used to build the formulas from the
data structures has done all the essential work in organizing the
items into a list.  Establishing these formulas
proves that the order and multiplicity of elements in the binary
tree and in the list in a provable traverse computation are the same.

\section{Future work}\label{future}

Various extensions of the basic scheme described here are natural to
consider.  In particular, it should be easy to consider approximating
data structures that contain items of differing types: each of these
types could be mapped into different $\hbox{\sl item}_\alpha(\cdot)$
predicates, one for each type $\alpha$.  

It should also be simple to construct approximating mappings given the
{\em polymorphic} typing of a given constructor's type.  For example,
if we are given the following declaration for binary tree (written
here in $\lambda$Prolog syntax),
\begin{verbatim}
kind btree   type -> type.
type emp     btree A.
type bt      A -> btree A -> btree A -> btree A.
\end{verbatim}
it should be possible to automatically construct the mapping
$$\halign{\hfill $# $ & $\mapsto #$\hfill\cr
  \Tok{btree}& \lambda x.\Tok{o}\cr
  \Tok{emp}  & \bot\cr
  \Tok{bt}   & \lambda x\lambda y\lambda z. \hbox{\sl item}_A(x)\lpar x\lpar y\cr
}$$
that can, for example, approximate a binary tree with the multiset of
the labels for internal nodes.

Abstract interpretation \cite{cousot77popl} can associate to a program
an approximation to its semantics.  Such approximations can help to
determine various kinds of properties of programs.  It will be
interesting to see how well the particular notions of collection
analysis can be related to abstract interpretation.  More challenging
would be to see to what extent the general methodology described here
-- the substitution into proofs (computation traces) and use of linear
logic -- can be related to the very general methodology of abstract
interpretation.

\acks

I am grateful to the anonymous reviewers for their helpful comments on
an earlier draft of this paper.  This work was funded in part by the
Information Society Technologies programme of the European Commission,
Future and Emerging Technologies under the IST-2005-015905 MOBIUS
project.  This paper reflects only the author's views and the
Community is not liable for any use that may be made of the
information contained therein.

\bibliographystyle{abbrv}

\end{document}